\documentclass[reprint, superscriptaddress, amsmath,amssymb, aps, prx, longbibliography, floatfix]{revtex4-2}
 
\UseRawInputEncoding

\usepackage{graphicx}
\usepackage{dcolumn}
\usepackage{bm}
\usepackage{braket}
\usepackage{physics}
\usepackage{xcolor}
\usepackage{comment}
\usepackage{amsthm}
\usepackage{soul}
\usepackage[utf8]{inputenc}

\newcommand{\cE}{\mathcal{E}}
\newcommand{\cF}{\mathcal{F}}

\newcommand{\cO}{\mathcal{O}}

\newcommand{\cR}{\mathcal{R}}

\definecolor{mypurple}{HTML}{7c4a6d} 
\definecolor{myred}{HTML}{4487ab} 
\definecolor{myblue}{HTML}{007A87} 
\definecolor{flamingo}{HTML}{c81455}
\definecolor{mydarkpurple}{HTML}{351c75}
\definecolor{royalblue}{HTML}{3791ff}
\definecolor{wine_stain}{rgb}{0.5,0,0}

\usepackage{hyperref}
\hypersetup{
    unicode = {true},
    colorlinks={true},
    linkcolor = {royalblue},
    filecolor = {violet},      
    urlcolor = {myblue},
    citecolor = {magenta},
    pdftitle = {Demonstrating Noise-adapted Quantum Error Correction With Break-Even Performance},
    pdfpagemode = {FullScreen},
    pdfauthor = {VAS},
    linktoc = {all},
    linktocpage = {true}
    }

\DeclareMathAlphabet{\mathbb}{U}{bbold}{m}{n}  
\global\long\def\idmatrix{\mathbb{1}}%

\newtheorem{lemma}{Lemma}

\begin{document}


\title{Demonstrating Noise-adapted Quantum Error Correction with Break-Even Performance}

\author{Vismay Joshi}
\thanks{Corresponding author: vismay@physics.iitm.ac.in}
\affiliation{Department of Physics, Indian Institute of Technology Madras, Chennai - 600036, India.}
 
\author{Anubhab Rudra}%
\affiliation{Department of Physics, Indian Institute of Technology Madras, Chennai - 600036, India.}

\author{Sourav Dutta}
\affiliation{Department of Physics, Indian Institute of Technology Madras, Chennai - 600036, India.}

\author{Siddharth Dhomkar}
\affiliation{Department of Physics, Indian Institute of Technology Madras, Chennai - 600036, India.}
\affiliation{Center for Quantum Information, Communication and Computing, IIT Madras, Chennai - 600036, India.}
\affiliation{London Center for Nanotechnology, University College London, London WC1H0AH, United Kingdom.}  

\author{Prabha Mandayam}
\affiliation{Department of Physics, Indian Institute of Technology Madras, Chennai - 600036, India.}
\affiliation{Center for Quantum Information, Communication and Computing, IIT Madras, Chennai - 600036, India.}


\begin{abstract}
The promise of quantum computing is closer to reality today than ever before, thanks to rapid progress in the development of quantum hardware. Even as qubit lifetimes and gate fidelities continue to improve, realizing robust, fault-tolerant quantum computers is contingent upon the successful implementation of quantum error correction (QEC). Conventional QEC schemes have rather high resource overheads and low threshold requirements, making them challenging to implement on present day hardware. Here, we use a recently developed noise-adapted $3$-qubit QEC scheme to demonstrate break-even performance against native amplitude-damping (AD) noise on IBM quantum hardware. We use variational quantum circuits to construct hardware-efficient encoding and decoding circuits. This scheme is probabilistic due to the non-unitary nature of the recovery operators, which are implemented via the block-encoding technique. We demonstrate logical qubit lifetimes exceeding those of the physical qubits by performing multiple rounds of QEC. To further protect the qubits from dephasing due to crosstalk, we incorporate dynamical decoupling into our noise-adapted QEC scheme in a seamless fashion. To account for the post-selection overhead, we define a measure of \textit{gain}, that allows for faithful performance benchmarking of the protocol. Our analysis suggests that the performance of our protocol is limited primarily by the measurement readout fidelity, and is bound to improve with successive generations of quantum processors.
\end{abstract}

\maketitle

\section{\label{sec:level1} Introduction}

Quantum processors have the potential to speed up complex computational tasks, leading to an exponential advantage in certain cases~\cite{google2019, morvan2024phase}. However, the quantum processors of today are noisy, with qubits that are inherently fragile and error-prone. Moreover, faulty quantum gates and measurement read-out processes are still far from ideal. The high levels of noise in current implementations is a crucial factor that prevents such noisy quantum devices~\cite{Preskill_2018} from reaching their full potential. One of the central challenges that theorists and experimentalists alike are grappling with today, is indeed the question of how to transition from the small, noisy quantum computing devices of today, to robust and scalable quantum computers~\cite{eisert2025, quantum_nogo}.



Quantum error correction (QEC)~\cite{LidarBrun_2013, Terhal_2015,  Campbell_2017} offers one such strategy to overcome the effects of noise on quantum bits, by encoding information into \emph{logical qubits} that are entangled states of the bare, \emph{physical qubits}. The prevailing QEC schemes can be broadly classified into (a) general-purpose schemes that can correct for arbitrary noise, and, (b) noise-adapted schemes~\cite{jayashankar2023} that are tailored to correct for specific noise processes. The well known stabilizer codes \cite{gottesman1997}, including topological codes like the surface codes~\cite{fowler2012surface}, fall under the former category. While the surface codes are attractive due to their planar structure and high error thresholds, they are rather resource intensive, with the shortest surface code requiring $17$ physical qubits to realise a single logical qubit~\cite{zhao2022realization}. 

Unlike the conventional QEC approach which aims to correct for independently occurring \emph{Pauli errors}, noise-adapted QEC schemes address the dominant noise process affecting the physical hardware. The framework of noise-adapted QEC therefore admits shorter codes that do not necessarily fit within the stabilizer framework, such as the $4$-qubit~\cite{Leung} and $4$-qudit codes~\cite{Dutta_qudit, AD_2} tailored to correct for amplitude-damping noise, or, cat codes~\cite{mirrahimi2014dynamically} and bosonic codes~\cite{Michael2016} that primarily correct for photon loss. Such codes also fall within the purview of \emph{approximate} QEC~\cite{beny, prabha}, which relaxes the rigid mathematical constraints of the well known Knill-Laflamme conditions for perfect QEC~\cite{KLCondition}. Another key point of departure is that noise-adapted schemes often require a non-trivial recovery operation~\cite{Fletcher, Debjyoti_NARC}, unlike for example, stabilizer codes where the syndrome extraction step is simply followed by a Pauli recovery operation. 

More recently, it was demonstrated that there exists a $3$-qubit code correcting for single-qubit amplitude-damping (AD) noise~\cite{Dutta_2024}, that uses the idea of probabilistic QEC~\cite{nayak2006, alhejji, pQEC2, pQEC3, pQEC4}. This code engenders a further relaxation of the QEC framework, where the recovery is implemented via post-selection. The probabilistic QEC scheme in~\cite{Dutta_2024} was shown to achieve a fidelity that is higher than all known codes for amplitude-damping noise, including the $4$-qubit Leung code and the $5$-qubit stabilizer code \cite{gottesman1997}.

Keeping pace with the theoretical progress in QEC, the past few years have witnessed several experimental demonstrations of QEC schemes on the current generation of noisy intermediate-scale quantum (NISQ) devices, often showing break-even performance. The first such demonstration of a logical qubit whose lifetime exceeded the physical $T_1$ lifetime was achieved using cat codes in superconducting cavities~\cite{Ofek_2016}. Since then, multiple hardware implementations of QEC codes have reported logical qubits with lifetimes surpassing the bare qubit $T_{1}$, including binomial codes~\cite{Ni_2023, Hu_2019}, trapped-ion logical qubits using the $[[7,1,3]]$ code~\cite{Paetznick_2024}, qubit and qudit GKP codes on superconducting qubits~\cite{Sivak_2023, Brock2025}, and, large-scale surface-code implementations~\cite{Google_2024}. 

In this work, we implement the $3$-qubit probabilistic quantum code on superconducting qubits provided by IBM Quantum, to demonstrate break-even performance in the logical qubit lifetime.  We construct explicit circuits for encoding and recovery using variational quantum circuits \cite{Cerezo2021}. 
An important feature of our implementation is that it combines a noise-adapted QEC scheme tailored for AD noise, with a dynamical decoupling (DD)~\cite{Viola_1999, Khodjasteh_2005, Uhrig_2007, Yang_2010} scheme in order to mitigate the effects of dephasing noise. Specifically, we make use of the recently developed chromatic Hadamard dynamical decoupling (CHaDD) scheme \cite{CHaDD} which schedules DD pulses sequentially by using Hadamard matrices, on arbitrary topologies, to mitigate dephasing noise and crosstalk on nearest neighbour qubits \cite{Buterakos_2018_Crosstalk_DD_Semiconductor, Tripathi_2022_Crosstalk_DD_Superconducting}.





We perform multiple cycles of QEC (multi-QEC) to preserve the logical $\ket{0_{L}}$, $\ket{1_{L}}$, and $\ket{+_{L}}$ states, thus providing a proof-of-principle demonstration of the ability to protect any arbitrary state with our scheme. For the logical $\ket{+_{L}}$, we concatenate multi-QEC with CHaDD so as to suppress the detrimental crosstalk noise. Since our scheme requires only $5$ physical qubits in total for a single logical qubit, we are able to demonstrate multiple logical qubits with lifetimes exceeding physical qubits, simultaneously on the same quantum processor.  

Finally, we provide a metric to quantitatively compare the performance of our probabilistic QEC scheme vis-a-vis the bare physical qubit lifetimes and illustrate that our protocol indeed offers a visible advantage in certain regimes. Owing to the post-selection of the successful instances inherent to the proposed protocol,  measurement errors emerge as the primary factor limiting the achievable gain.

The rest of the manuscript is organized as follows. In section \ref{sec:level2}, we describe the $3$-qubit code and the circuit designs for the encoding, syndrome extraction and the recovery. In section \ref{sec:level3}, we discuss the results obtained from simulations and the implementation on actual hardware. In Section \ref{sec:level4}, we discuss the implications of our results and outline prospective directions enabled by advancements in hardware. 

\section{\label{sec:level2} Towards Physical Realization of Probabilistic Error Correction}
\subsection{The $3$-qubit code for amplitude-damping noise}\label{sec:code}


A qubit, starting in the excited state, naturally relaxes to the ground state while emitting a photon. This physically motivated loss mechanism, called amplitude-damping (AD) noise ~\cite{AD_1, AD_2} channel, arises via a Jaynes--Cummings interaction~\cite{JC_interaction} between a qubit and its bath. The action of the AD noise channel $\mathcal{A}$ on a single qubit can be described in terms of two \emph{Kraus operators}, labeled $A_0$ and $A_1$, which correspond to the no-damping error and the single-qubit damping error respectively, as described below.
\begin{align}
    A_0 = \ket{0}\bra{0} + \sqrt{1-\gamma} \ket{1}\bra{1}, \qquad A_1 = \sqrt{\gamma} \ket{0}\bra{1}.
    \label{eq:kraus}
\end{align}
Unlike unitary Pauli noise, AD is inherently asymmetric and non-unitary: excitations can be lost, but they cannot be spontaneously regenerated. The error probability or noise strength $\gamma$ can be empirically modeled as a function of time as $ \gamma(t) = 1 - e^{-\frac{t}{T_{1}}}$, where the time-constant $T_{1}$ is the characteristic lifetime or the \emph{relaxation time} of the qubit. 
Such a noise model is relevant for systems with bosonic cavities, such as trapped ions and superconducting qubits \cite{Astafiev_2004_Josephson_Quantum_Noise, Rob_2002_Qubits_Spectrometers_noise, Krantz_2019_quantum_engineer, Tomita_2014_Surface_realistic_noise, Gicev_2024_error_structure}.

We consider the $3$-qubit code introduced in ~\cite{Dutta_2024}, tailored to protect against single-qubit AD noise. This code encodes a single logical qubit into fixed excitation, permutation invariant states, also known as the Dicke states \cite{Dicke}. The logical states (codewords) are given by,
\begin{equation}\label{eq:three_qubit_code}
    \ket{0_L} = \frac{1}{\sqrt{3}}(\ket{100}+\ket{010}+\ket{001}) ~~\text{and} ~~ \ket{1_L} = \ket{111}
\end{equation}
It can be easily checked that the no-damping and single-damping errors map the codewords in Eq.~\eqref{eq:three_qubit_code} to orthogonal subspaces, thus enabling us to distinguish between the two cases. However, this code does not allow for uniquely identifying the qubit that experienced a single damping event, marking a deviation from the standard QEC approach for syndrome extraction.

The syndrome extraction step involves a single $Z_{1} Z_{2} Z_{3}$ measurement, where $Z_{i}$ is the Pauli $Z$ on the $i^{\rm th}$ qubit. Although this code does not have a complete
set of stabilizer generators, the operator $-Z_{1} Z_{2} Z_{3}$ stabilizes the codewords. 
If the outcome of the $Z_{1} Z_{2} Z_{3}$ measurement is $+1$, one can conclude that a single damping error has occurred, whereas a $-1$ outcome indicates the occurrence of the no-damping error. Upon successful detection of the error syndrome, a carefully designed non-unitary recovery map $\cR$ restores the logical state with high fidelity.

\begin{figure*}[t]
    \centering
    \includegraphics[width=1.0\linewidth]{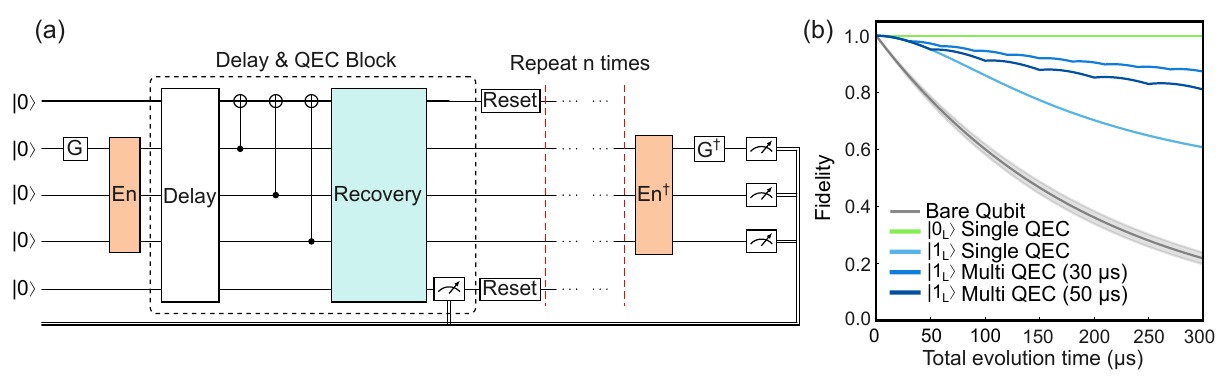}
    \caption{
    {\bf (a) Schematic circuit for multi-cycle QEC.} The unitary $G$ and the encoder {\bf($\textrm{En}$)} prepare an arbitrary logical state, that undergoes both damping and dephasing during the free evolution {\bf (Delay)}. Syndrome extraction and recovery each require an additional ancilla. The delay \& QEC block is repeated for $n$ rounds, resetting the ancilla each time. Finally, $(\textrm{En})^{\dagger}$ and $G^{\dagger}$ are applied, followed by a measurement to estimate the fidelity. {\bf (b) Fidelity of the physical and logical qubits as a function of time.} The total evolution time is the sum of the delay and QEC protocol runtimes. The grey line (background) is the average $T_{1}$ ($T_{1}$ spread). The green (blue) line is for $\ket{0_{L}}$ ($\ket{1_{L}}$). Multi-QEC is plotted for $\ket{1_{L}}$ with maximum allowed delays of $30 \mu s$ and $50 \mu s$ between consecutive QEC cycles.}
    \label{fig:1}
\end{figure*}

The performance of this QEC protocol for an arbitrary logical state $\ket{\psi_L} = \cos{\frac{\theta}{2}} \ket{0_L} + e^{i\phi} \sin{\frac{\theta}{2}} \ket{1_L}$ can be benchmarked using the fidelity function
$\cF_{\ket{\psi_L}} = \matrixel{\psi_L}{(\cR\circ\cE)(|\psi_{L}\rangle\langle\psi_{L})}{\psi_L}$, which evaluates to~\cite{Dutta_2024}, 
\begin{equation}\label{eq:statefid}
    F_{\ket{\psi_L}} = \frac{1+\gamma^2 \sin^2{(\frac{\theta}{2})} \cos^2{(\frac{\theta}{2})} }{1+\gamma^2 \sin^2{(\frac{\theta}{2})}}.
\end{equation}
For a given noise strength $\gamma$, the fidelity is minimum for the state $\ket{1_L}$ (when $\theta = \pi$), thus yielding a \emph{worst-case fidelity} of,
\begin{equation*}
    F_{\text{worst-case}} = \frac{1}{1+\gamma^2} = 1 - \gamma^2 + \mathcal{O}(\gamma^3).
\end{equation*}
The recovery channel $\cR$ that achieves this fidelity is a trace non-increasing map described by two operators, corresponding to the no-damping and single-damping errors respectively~\cite{Dutta_2024}.
\begin{eqnarray}
    R_0 &=& (1-\gamma) \dyad{0_L} + \dyad{1_L}, \label{eq:recovery} \\
    R_1 &=& (1-\gamma) \dyad{0_L}{000} + \frac{1}{\sqrt{3}} \ket{1_L}(\bra{011} + \bra{101} + \bra{110}). \nonumber
\end{eqnarray}
The recovery is indeed noise-adapted with both operators carrying an explicit dependence on the noise strength $\gamma$. Using known circuit constructions for implementing a non-unitary operators~\cite{non_unitary}, it was shown in~\cite{Dutta_2024} that the recovery in Eq.~\eqref{eq:recovery} can be implemented using a post-selected scheme. The non-unitary nature of the recovery makes this a probabilistic QEC scheme, with a probability of success given by,
\begin{align}\label{eq:prob_succ}
    p_{success} = (1-\gamma)^2(1-\gamma^2 \sin^2{\frac{\theta}{2}}).
\end{align}


\subsection{The $3$-qubit code with dephasing and damping}\label{sec:dephasing}

In superconducting qubits, AD is accompanied by dephasing noise as well. The single-qubit pure dephasing channel is composed of the Kraus operators $\{\sqrt{1-p }I, \sqrt{p}Z\}$, where $I$ denotes the identity operator and $Z$ denotes the Pauli operator that is diagonal in the computational basis. Here, $p = \frac{1}{2}(1 - \exp(-t/T_{\phi}))$ is the probability of a dephasing error and $T_{\phi}$ denotes the \emph{dephasing lifetime}. A combination of $T_{1}$ and $T_{\phi}$ gives phase-damping or $T_{2}$ noise.  

Interestingly, the recovery channel $\cR$ protects the codewords from the first-order dephasing noise, so their fidelities under the action of the combined noise channel and recovery, are given by $F_{\ket{0_{L}}} = 1$ and $F_{\ket{1_{L}}} = 1 - \gamma^{2} + \cO(\gamma^{3})$.
However, for superposed logical states, the fidelity drops further when one includes the dephasing noise along with AD, as seen in Fig.~\ref{fig:1}. Specifically, the $3$-qubit QEC scheme protects the logical states $|+_{L}\rangle$ and $|-_{L}\rangle$ with fidelities of 
\begin{eqnarray}
&& F_{\ket{+_{L}}} (\gamma, p) \; = \;F_{\ket{-_{L}}} (\gamma, p) \nonumber \\
&=&
1-\frac{(4 + \gamma) p}{3}-\frac{4 p^2}{9} - \frac{5\gamma^2}{16} + \text{h.o.}
\end{eqnarray}
Here, h.o. stands for higher order terms. The presence of dephasing also affects the expression for the success probability of implementing the post-selected recovery. The analytical expressions for arbitrary state fidelity and  success probability achieved by the $3$-qubit QEC scheme in the presence of both damping and dephasing is given in Appendix \ref{sec: appA}. 



\begin{figure*}[t]
    \centering
    \includegraphics[width=1.0\linewidth]{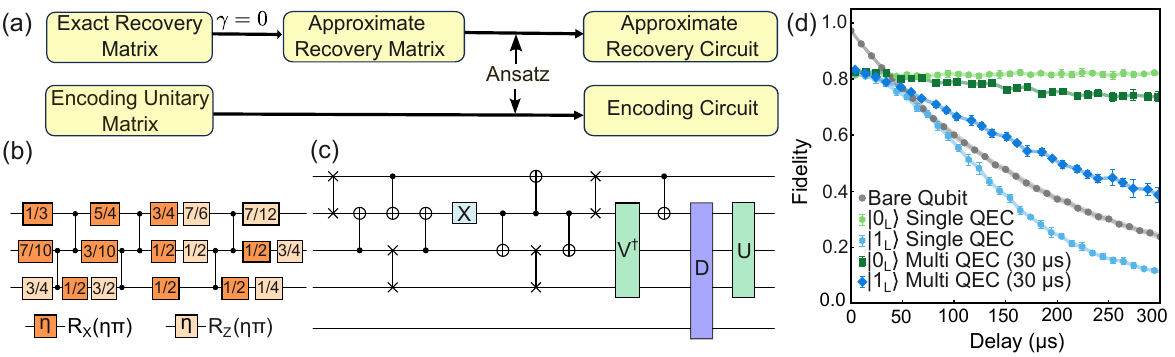}
    \caption{\textbf{The multi-QEC circuit design}. (a) \textbf{VQC and circuit design process flowchart}. We approximate the recovery by fixing  $\gamma = 0$. We use an ansatz and run a classical optimization routine to find the optimal hardware-efficient circuit decomposition. (b) \textbf{The encoder circuit}. The ansatz for all our QEC components has a similar structure as seen in the decomposed encoding circuit. $R_X(\theta)$ is denoted in orange, and $R_Z(\theta)$ in peach. The displayed value in the box $\eta$ represents the optimized rotational angle as multiples of $\pi$. (c) \textbf{The recovery circuit}. The syndrome measurement circuit ends before the $X$ gate in the subfigure. The explicit circuits are given in the Appendix \ref{sec: appB}. (d) \textbf{Approximate recovery performance}. The plot shows the simulated fidelity obtained with the approximate recovery, under $T_{1}$, $T_{2}$, gate, and readout noise. The grey dots give the average $T_{1}$. The light (dark) green (blue) points show the performance of the single (multi) QEC for $\ket{0_{L}}$ ($\ket{1_{L}}$).}
    \label{fig:2}
\end{figure*}  

\subsection{The multi-cycle QEC protocol}\label{sec:multi}

The schematic circuit 
for a single round of QEC using the $3$-qubit noise-adapted code is shown in Fig.\ref{fig:1}(a). The scheme requires only one  ancillary qubit in principle, apart from the three data qubits, which maybe reset and reused after syndrome extraction step to implement the post-selected recovery. Our implementation however makes use of two ancillary qubits, one each for syndrome extraction and the recovery respectively, for ease of implementation on physical hardware, as shown in Fig.~\ref{fig:1}. We first prepare an arbitrary single-qubit quantum state 
using the unitary $G$, and encode it in the codewords of the $3$-qubit code using the encoding circuit ($\textrm{En}$). We then allow for free evolution of the encoded qubits and ancilla qubits, denoted as (${\rm Delay}$) in Fig.\ref{fig:1}(a), during which time the qubits may undergo both damping and dephasing. Note that, in our protocol, {\rm Delay} refers to the free evolution of the qubits without any artificial injection of noise. This is followed by the $Z_{1} Z_{2} Z_{3}$ syndrome measurement using one ancillary qubit and the probabilistic recovery circuit which uses one additional ancillary qubit. The encoding and recovery circuits are described in detail in Sec.~\ref{sec:circuits}.

A single round of QEC after a long computation or delay, may not be able to restore the state properly. Hence, multiple rounds of QEC are performed at regular intervals to continually detect and correct errors before they accumulate beyond the error correction capability of the code. Thus, we repeat the delay and QEC block for a certain number ($n$) of rounds, depending on the total runtime of the circuit and the noise strength. 


The fidelity is measured at the end by applying $({\rm En})^{\dagger}$ and $G^{\dagger}$ and measuring the three encoded qubits as well as the ancilla used to perform the recovery, in the computational basis. We store the value of the ancilla after each round to post-select all the successful instances of the recovery. The probability of the all-zero outcome directly gives the fidelity. A detailed analysis of why the all-zero outcome corresponds to the state fidelity is given in the Appendix \ref{Appendix: fidelity}.

The plots in Fig.\ref{fig:1}(b) show how the fidelities of the physical and logical qubits drop off as a function of the delay. The fidelities have been calculated under an idealized setting, where the qubits are subject to $T_{1}$ and $T_{2}$ noise, but the hardware is assumed to be free of gate and measurement errors. The total evolution time in Fig.~\ref{fig:1} (b) includes the time taken for encoding, $n$ rounds of the delay and QEC block, followed by the inverse of the encoding. During the multi-QEC protocol, the total free evolution is the sum of the delays in the $n$ rounds. 

For each such round, the delay does not exceed the pre-specified maximum allowed value. In Fig.~\ref{fig:1} (b), we plot the fidelities for two delay values, $30 \mu s$ and $50 \mu s$, respectively. For example, if the total free evolution time is $40\mu s$, with the maximum delay of $30 \mu s$, then there will be two rounds of QEC. In the first round, the delay is $30 \mu s$, and $10 \mu s$ in the second. Our encoding circuit has a runtime of about $0.548 \mu s$, and the recovery circuit requires about $3.072 \mu s$. So, the total evolution time in this case would be $47.24 \mu s$. The relevant gate and reset times have been taken from IBM Quantum devices. The reset of the ancilla, requiring $2.72 \mu s$, occurs during the delay of the logical register before the start of the next QEC round, which adds no additional time to the full evolution.

\subsection{Encoding and Recovery Circuits}\label{sec:circuits}

We next present the details of the encoding, syndrome extraction and recovery circuits for our QEC protocol. We use a variational quantum circuits (VQC) approach~\cite{VQE-1_Peruzzo2014, VQE-2_Cerezo2021, McClean2016, Cao_2022} to efficiently generate low-depth quantum circuits for the encoding and recovery operations. A flowchart highlighting the important steps in the assembly of the circuit designs is given in Fig.~\ref{fig:2} (a). 

The codewords of the $3$-qubit code are indeed easy to prepare, with the $|0_{L}\rangle$ being the well known $W$-state \cite{W-Dur_W_state} and the $|1_{L}\rangle$ being a simple product state. While there exist standard circuit preparations for the $W$-state \cite{W-state_expt, W_state_IBM}, our VQC-based circuit allows for the preparation of arbitrary logical states, while taking into account the hardware's native gate set as well as connectivity constraints.

For obtaining the hardware-optimized encoding circuit, we first choose an ansatz -- a parametrized circuit --  similar to the structure of Fig. (\ref{fig:2}b), with alternating layers of single-qubit $R_{X}(\theta)$, $R_{Z}(\theta)$ gates, and two-qubit $CZ$ gates acting between neighboring qubits. The gate $R_{\sigma_j}(\theta) = e^{i\sigma_j\theta/2}$ rotates the state of the qubit by $\theta$ angle around the $\sigma_j$-axis of the Bloch sphere, where $\sigma_j$ can be the Pauli matrices $X,Y$ or $Z$. The single-qubit rotation angles $\theta$ are the parameters we wish to optimize for, such that we end up with a circuit that is closest to the required encoding operation. We now define an associated cost function, $C(\boldsymbol{\vec{\theta}})$, which is a function of the vector $\boldsymbol{\vec{\theta}}$ of rotation angles used in our ansatz. 
The parameters are updated iteratively using classical optimization tools, given by SciPy \cite{2020SciPy-NMeth}, in order to minimize $C(\boldsymbol{\vec{\theta}})$.

We denote our ansatz as $U(\boldsymbol{\vec{\theta}})$ and the unitary matrix which we want to implement as $O$. In cases where only a part of matrix $O$ is known or only certain elements of $O$ are of interest, we can reduce the complexity of $C(\boldsymbol{\vec{\theta}})$ by summing over only relevant matrix elements in the expression below. 

\begin{align}\label{eq:8}
    C(\boldsymbol{\vec{\theta}}) = \sum_{i,j} |O_{ij} - U_{ij}(\boldsymbol{\vec{\theta}})|^2
\end{align}
In our encoding circuit, only the first two columns of the matrix $O$ are important. We will use only those two columns to optimize our choice of ansatz in order to get the final encoding circuit, given in Fig. \ref{fig:2} (b).

We implement the recovery map for the $3$-qubit code by obtaining circuit implementations of the non-unitary operators in Eq.~\eqref{eq:recovery}. In general, non-unitary operators can be implemented by either performing polar decomposition (PD)~\cite{non_unitary, Dutta_2024} or singular value decomposition (SVD)~\cite{non_unitary}. These decompositions provide an elementary sequence consisting of a unitary part and a non-unitary part which is implemented via the technique of block-encoding~\cite{Low2017}. The $3$-qubit non-unitary operators are implemented as a $4$-qubit unitary operation, using an extra ancilla qubit, followed by a measurement of the ancilla qubit. We post-select over the favorable outcomes of the ancilla and discard the rest, thus making this a probabilistic recovery process. 


\begin{figure}[t!]
    \centering
    \includegraphics[width=0.95\linewidth]{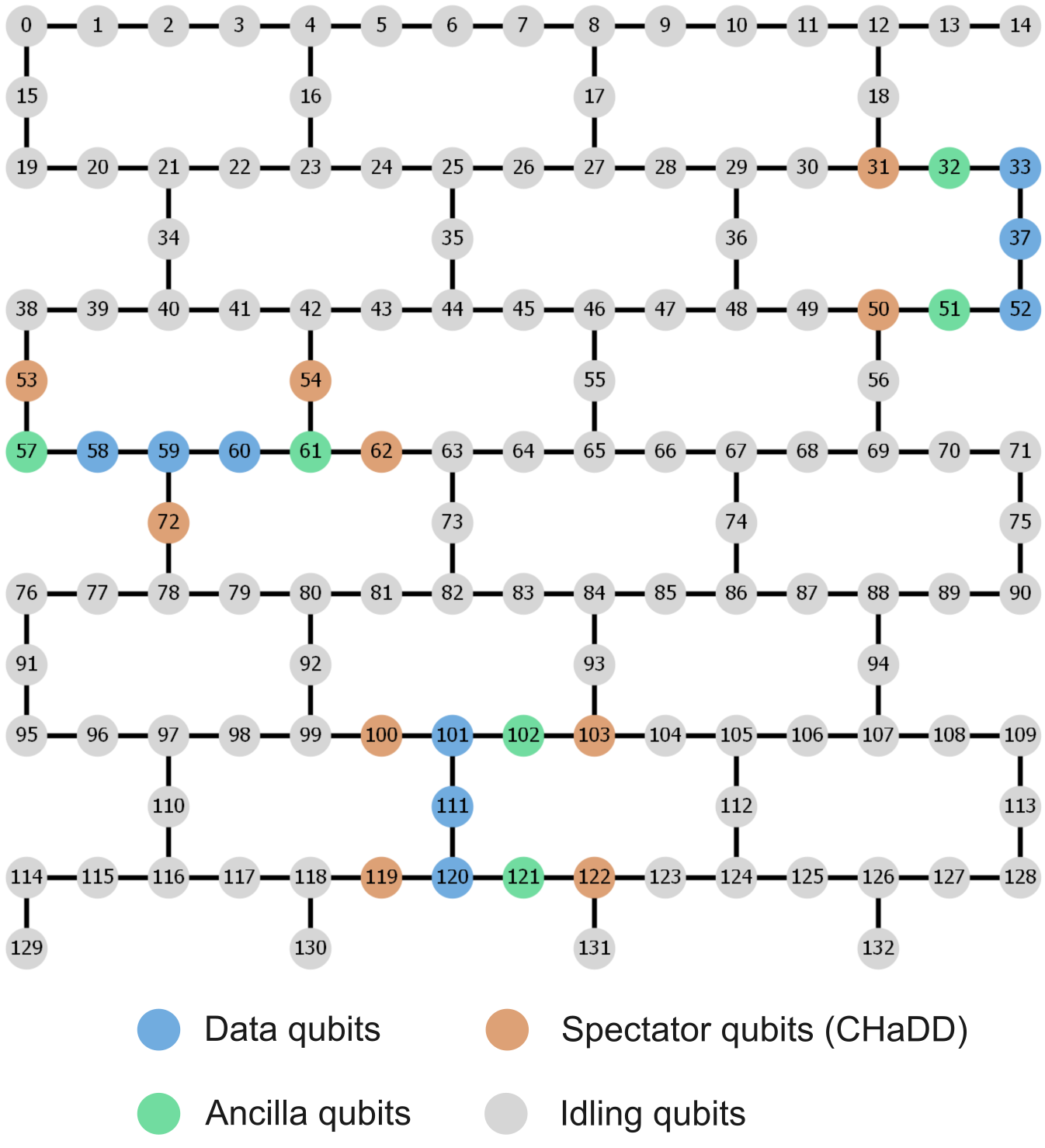}
    \caption{\textbf{Layout of the IBM Torino device}. Data qubits forming the logical qubit are shown in blue, ancilla qubits used for syndrome extraction and post-selection in green, and spectator qubits (nearest neighbour of the data qubits and the ancilla qubits) participating in the CHaDD sequence in orange. Idle qubits are marked in gray.}
    \label{fig:torino}
\end{figure}

For our recovery implementation, we choose the SVD route, since it has a distinct advantage over the PD approach. PD breaks the recovery into two parts, whereas the SVD breaks it into three parts. It is easier for VQC to optimize for smaller chunks or subroutines. Furthermore, in SVD, we have a diagonal part, which symbolizes more symmetry than PD. Recall that for a general matrix $R$, the singular value decomposition is of the form $R = UDV^{\dagger}$. 
Here, $D$ being the diagonal matrix of singular values, gives rise to the non-unitary part; the unitary $U$ ($V$) forms the eigenvector matrix of $RR^{\dagger}$ ($R^{\dagger}R$).

For the two recovery operators $\{R_{0}, R_{1}\}$ in Eq.~\eqref{eq:recovery}, the corresponding $U$ and $D$ are identical, that is, $U_{0} = U_{1}$ and $D_{0} = D_{1}$. For $R_{0}$, $V_{0} = U_{0}$, and for $R_{1}$, $V_{1} = X_{1}X_{2}X_{3}U_{1}X_{1}$, where $X_{i}$ is the Pauli $X$ operator on the qubit $i$. The qubit numbering order is the same as in Fig. (\ref{fig:1}a), where the topmost qubit begins with index $0$. 

\begin{figure*}[t!]
    \centering
    \includegraphics[width=1.0\linewidth]{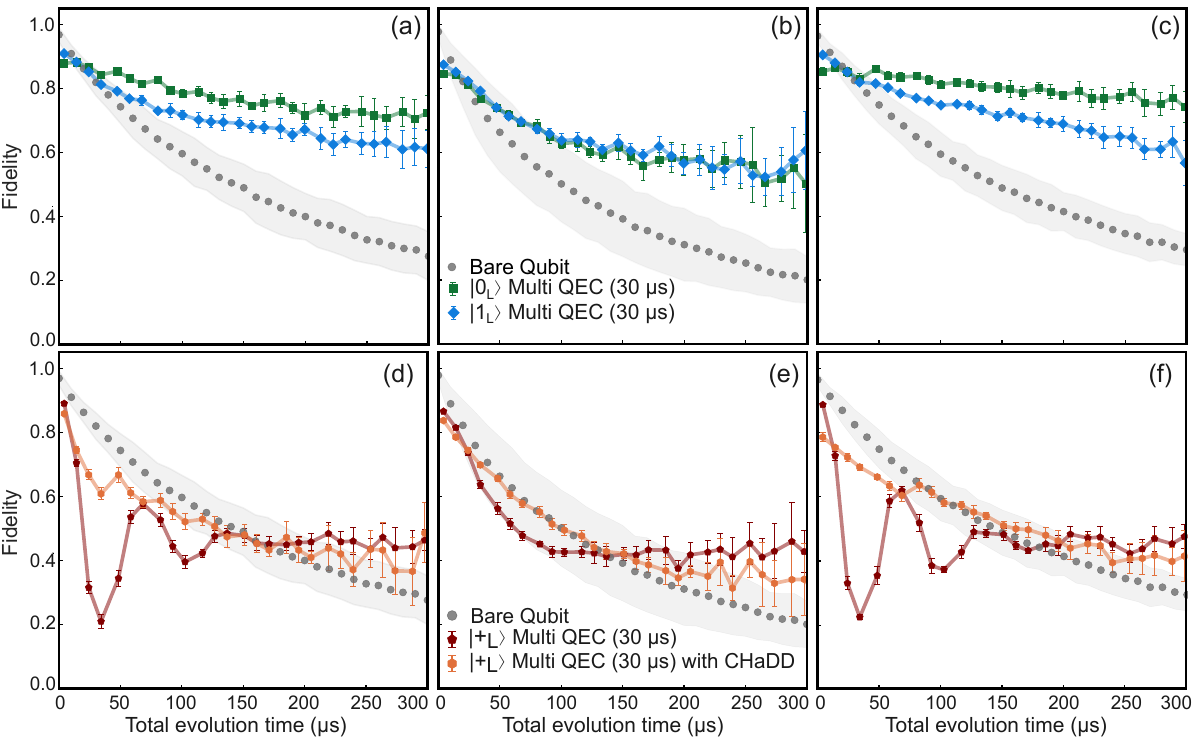}
    \caption{\textbf{Hardware experiment results of multi-QEC on IBM Torino}. Top three figures (a), (b), (c) show the multi-QEC logical qubit fidelity for the $\ket{0_{L}}$ and $\ket{1_{L}}$ states in comparison with the bare qubit. The grey dots (background) represent the average $T_{1}$ ($T_{1}$ spread) of the bare qubits of that set. The green (blue) squares (diamonds) emphasize the logical $\ket{0_{L}}$ ($\ket{1_{L}}$) fidelity.
    Bottom three figures  (d), (e), (f) show the multi-QEC logical qubit fidelity for the $\ket{+_{L}}$ state in comparison with the bare qubit, both with and without the CHaDD protocol. The grey dots (background) represent the average $T_{1}$ ($T_{1}$ spread) of the bare qubits of that set. The maroon (orange) pentagons (hexagons) show the fidelity of the logical $\ket{+_{L}}$ without CHaDD (with CHaDD).}
    \label{fig:3}
\end{figure*}

Having performed SVD, we now use VQC to find the circuits for only $U$ and $D$. To decompose $U$ and $V^\dagger$, we will only focus on the columns of $U$ and $V$ that correspond to the non-zero entries of $D$. The non-unitary diagonal operator $D$ can be realized using $C_{i}R_Y(\theta)$ gates, which are multi-controlled $R_{Y}(\theta)$ gates, where $i$ denotes the control on $i$ qubits. In this construction, the ancilla qubit introduced for the block encoding serves as the target, while the data qubits act as controls. Since the non-zero diagonal elements of $D$ are $1$ and $1 - \gamma$, two $C_{3}R_{Y}$ gates are needed for exact implementation. However, the decomposition of such gates requires more than $60$ $CZ$ gates \cite{cnot_scaling_Rakyta2022}, leading to long execution times and significant error accumulation. 

We apply $R_{0}$ or $R_{1}$ by applying syndrome-dependent control operations. One can either measure the syndrome qubit and apply the corresponding recovery based on the output via classical feedback or directly use local quantum controlled operations, as in our case.For, the classical-feedforward based quantum circuit please refer to the GitHub in the Data Availability section.

To reduce overhead, we approximate $D$, by setting $\gamma = 0$. The fidelity with this approximate recovery under the joint action of AD and dephasing noise is,
\begin{align}
    & F_{\ket{0_{L}}}= 1; \; F_{\ket{1_{L}}} =  1 - \gamma^{2} + \text{h.o.}; \nonumber\\
    & F_{\ket{+_{L}}} = 1-\frac{4 p}{3} - \frac{4 p^2}{9} - \frac{\gamma^{2}}{2} + \text{h.o.} 
\end{align}
For $p = 0$, the approximate recovery still accounts for all first-order damping. The implementation is now simpler due to $\gamma$-independence. In this case, both diagonal elements become $1$, and the operation reduces to a single $C_{2}R_{Y}(\pi)$ gate. This gate can be efficiently implemented using a `Margolus-type' ansatz \cite{Margolus}, requiring only $5$ $CZ$ gates. The explicit circuits for $U$ and $D$ are provided in the Appendix \ref{sec: appB}.  

The interval between two QEC cycles is chosen so that $\gamma$ remains small, ensuring the validity of the approximation used in the recovery operation. An analysis of the optimal maximum allowed delay is made in Appendix \ref{appendix:C}. We have observed that $30 \mu s$, corresponding to $\gamma \approx 0.14$ for a qubit with $T_{1} \simeq 200 \mu s$ is the required sweet spot for IBM's qubits.

\section{\label{sec:level3} Experimental Results \& Discussion}
We now present the results of our implementation of the $3$-qubit QEC protocol on IBM Quantum hardware. We demonstrate break even performance for three different sets of logical qubits on the IBM Torino device~\cite{IBM_Torino}, the layout of which has been shown in Fig.~\ref{fig:torino}. We present the data for three different sets of logical qubits on which the multi-QEC protocol outlined in Sec.~\ref{sec:multi} was implemented. The experiment was carried out simultaneously on the different sets.

\subsection{Demonstration of logical $|0_{L}\rangle$ and $|1_{L}\rangle$ on the hardware}

The multi-QEC protocol executed here has a maximum delay $30 \mu s$ between two QEC cycles. Fig. \ref{fig:3} (a, b, c) shows the performance of the protocol, in terms of the fidelity, for the logical states $\ket{0_{L}}$ and $\ket{1_{L}}$ with the total free evolution time. The free evolution time includes the time for each recovery process, which takes $3.072 \mu s$. The bare qubit and logical fidelities do not start from $1$ as there are readout errors present which are of the order of $10^{-2}$. 

The total shots for each data point for the logical qubit fidelity is $52000$, and the corresponding number for the bare qubit $T_{1}$ curve is $2000$. We have observed that on an average the data follows the trend of the noisy simulation, where the multi-QEC fidelities of both $\ket{0_{L}}$ and $\ket{1_{L}}$ are greater than the bare-qubit fidelity. In fact, the logical qubit lifetimes comfortably beat the average bare qubit $T_{1}$ lifetimes of $220 \mu s$, with an estimated value of $740 \mu s$ for the average logical qubit $T_{1}$ over the three sets. The standard deviation of the fidelity for the logical states increases gradually. This is attributed to the fact that the success probability of the scheme decreases very rapidly as seen from the expression given in the Appendix \ref{sec: appA}.

\subsection{Demonstration of logical $|+_{L}\rangle$ on the hardware: the CHaDD protocol}

Fig. \ref{fig:3} (d, e, f) illustrates the fidelity of the $\ket{+_{L}}$ for the aforementioned sets of logical qubits. As our protocol is not meant to protect the coherence or improve the $T_{2}$ lifetimes, the decrement in the logical qubit's fidelity is rapid in the case of $\ket{+_{L}}$. Moreover, the presence of unintentional crosstalk between the qubits manifests itself in the form of visible oscillations. The qubits investigated in \ref{fig:3} (d, f) exhibit relatively high crosstalk frequency, whereas the fidelity in (e) decays to $0.5$ before we can observe the characteristic oscillations due to a much lower crosstalk frequency. The main issue here arises from the nature of the crosstalk, which can either be static or can appear intermittently while the nearby qubits are being driven. However, with the help of the interleaved CHaDD \cite{CHaDD} protocol, we are able to mitigate this problem to a large extent. 

CHaDD is a colour-aware dynamical decoupling scheme which takes into account the chromaticity ($\chi$) of the underlying qubit connectivity graph. The chromaticity refers to the minimum number of colours required to colour the vertices of the graph such that no two adjacent vertices share the same colour. IBM Torino has the heavy-hex architecture with $\chi = 2$, which means we assign two different colours to the adjacent qubits. The sequence of the CHaDD pulse is determined from a `\textit{sign}' matrix which is $H^{\otimes 2}$ for our case; $H$ being the Hadamard matrix.
We arbitrarily choose two rows of the sign matrix, except for the first row, and assign them the qubits' colours.
The columns of the sign matrix, from left to right, represent the time steps of the DD, which helps to schedule the pulses. 
In a column, a positive or a negative entry corresponds to the free evolution of the qubits or applying an $X$ gate on the qubits with the same colour, respectively, which are mapped to the corresponding row. Due to the orthogonality of these rows, the effective DD sequence suppresses the $ZZ$ crosstalk and dephasing up to first order in $\tau$, which is the free evolution period of the qubits.

Here, the gate time is considered to be very short compared to the free evolution period. IBM devices have gate times in the order of tens of $ns$ and free evolution is of the order of $\mu s$, which validates the underlying assumption. In our protocol, we use the robust variant of the above-mentioned single-axis $X$-type CHaDD, with chromaticity $2$. We call a qubit a spectator if it does not participate in the main QEC protocol. In our scheme, we consider all the nearest-neighbour qubits as spectators, as shown in Fig. (\ref{fig:torino}). The DD sequence used in our protocol has the following form:
\begin{align*}
    X_1 f_\tau X_2 f_\tau \widetilde{X}_1 f_\tau \widetilde{X}_2 f_\tau \widetilde{X}_1 f_\tau \widetilde{X}_2 f_\tau X_1 f_\tau X_2
\end{align*}
where $f_\tau$ represents the free evolution unitary of the overall system, for a duration $\tau$, and $X_i$ ($\widetilde{X}_{i}$) represents the $X$ ($R_{X}(-\pi)$) operation being applied to all qubits of colour $i$.


\begin{figure}
    \centering
    \includegraphics[width=1.0\linewidth]{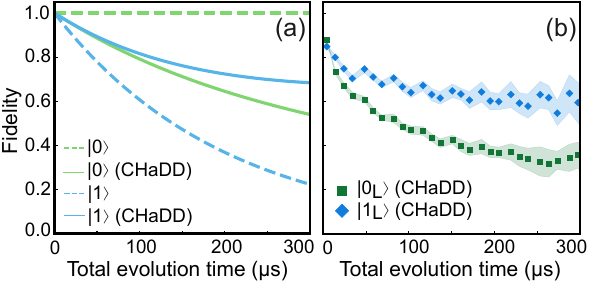}
    \caption{\textbf{Analyzing the effect of CHaDD on logical qubits}. a) \textbf{Toy model simulation}. A toy $2$-qubit model, with $ZZ$ crosstalk interaction, with populations of $\ket{0}$ (green) and $\ket{1}$ (blue) of the probe qubit is shown. The solid line shows the populations with the application of the CHaDD protocol. b) \textbf{ChaDD executed during multi-QEC for the codewords}. Shows the effect of CHaDD on the logical qubit's $\ket{0_{L}}$ and $\ket{1_{L}}$ populations averaged over the same $3$ qubit sets considered previously. The background shaded area is the standard deviation across the different sets.}
    \label{fig:4}
\end{figure}

\begin{figure*}
    \centering
\includegraphics[width=1.0\linewidth]{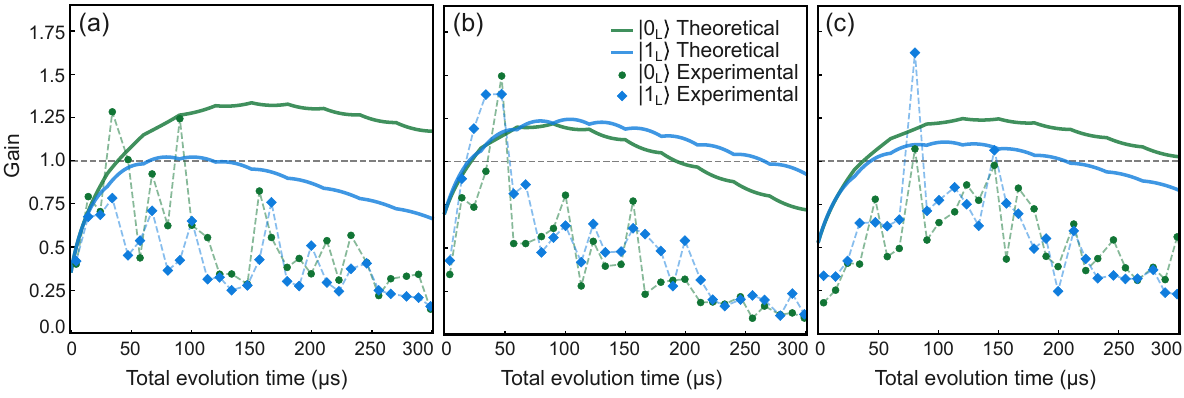}
    \caption{\textbf{Performance gain of the multi-QEC protocol}. In (a), (b), (c) the green (blue) circles (diamonds) represent gain values for the logical $\ket{0_{L}}$ ($\ket{1_{L}}$) state, obtained using Eq. \eqref{eq:gain_expt}. During multi-QEC to protect the codewords we do not use the CHaDD scheme. The corresponding solid lines show the behaviour expected from the theoretical model as given in Eq. \eqref{eq:gain}. It captures the trend of the peak of the experimental gain, but fails in other regimes, which can be attributed to gate errors and crosstalk.}
    \label{fig:gain}
\end{figure*}

\subsection{Analyzing the effect of CHaDD on $\ket{0_L}$ and $\ket{1_L}$}

Simultaneous implementation of QEC and DD to stabilise an arbitrary qubit state has not been explored previously. Although CHaDD proves beneficial while working with $\ket{+_{L}}$ by readily suppressing the crosstalk, it has a non-trivial impact on the of $\ket{0_L}$ and $\ket{1_L}$. In order to analyse the effect of CHaDD on our QEC scheme, we consider a $2$-qubit toy model that includes relaxation, dephasing, and static $ZZ$ crosstalk interaction given by the Hamiltonian: 

\begin{align}
    H/ \hbar = \frac{\omega_{1}}{2} Z_{1} + \frac{\omega_{2}}{2} Z_{2} + g Z_{1}Z_{2}
\end{align}

Here, $\omega_{1}$ and $\omega_{2}$ are the qubit frequencies, $g$ is the coupling strength, and $Z_{i}$ is the Pauli Z on the $i$-th qubit. We perform $4$ independent simulations using the Qutip package \cite{Qutip}; the results are as shown in Fig. \ref{fig:4}(a). We initialise the probe qubit in either $\ket{0}$ or $\ket{1}$, while fixing the initial state of the second qubit to $\ket{0}$. We allow the system to evolve in the presence or absence of the CHaDD pulses, subsequently recording the state fidelity of the probe qubit. It can be seen from Fig. \ref{fig:4}(a) that the fidelity for $\ket{0}$ drops with CHaDD while it improves for $\ket{1}$. Intuitively, this observation can be explained in the following manner. The CHaDD protocol is essentially a sequence of $X$ gates on the two qubits. If one applies an $X$ gate on $\ket{0}$, then the population will be transferred to $\ket{1}$, which is prone to relaxation. Similarly, for $\ket{1}$, if we apply an $X$ gate, then that results in the partial suppression of relaxation. This is essentially a consequence of the inherent asymmetry of the AD channel. 

The toy model presented here faithfully mimics the dynamics of the superconducting qubits, which operate at temperatures of about $10mK$. In such qubits, the damping contribution to $T_{1}$ is far greater than that of the thermal excitations. Thus, the behaviour of the logical qubit, as plotted in Fig. \ref{fig:4}(b), is in very good agreement with the toy model predictions. Here, the traces represent the average value of state fidelity across the three sets of qubits.

\subsection{Quantifying the advantage}

Although we clearly demonstrate that our protocol boosts $T_{1}$ of the logical qubits significantly, the probabilistic nature of the protocol necessitates a more nuanced proxy to quantify the practical advantage. Similar to the performance metrics that are suitable for multi-shot approaches, we develop a measure -- \textit{Gain} -- that relies on the effective signal-to-noise ratio (SNR). Thus, we define experimental \textit{Gain} as follows:

\begin{equation}\label{eq:gain_expt}
    \text{Gain}_{expt} = \frac{\text{SNR}_{\text{QEC}}}{\text{SNR}_{\text{bare}}} = \frac{F_{\text{QEC}} \sqrt{N_{\text{bare}}} \sigma_{\text{bare}}}{F_{\text{bare}}\sqrt{N_{\text{QEC}}} \sigma_{\text{QEC}}}
\end{equation}
Here, ${\text{SNR}_{\text{QEC}}}$ represent SNR corresponding to the QEC protocol, $F_{\mathrm{QEC}}$ denotes the fidelity after successful recovery with QEC, $N_{\mathrm{QEC}}$ is the total number of QEC shots, including the unsuccessful runs, and $\sigma_{\mathrm{QEC}}$ is the standard deviation of the fidelity data. Similarly, $F_{\mathrm{bare}}$ denotes the bare-qubit fidelity extracted from standard $T_1$ measurements, obtained by preparing the qubit in $\ket{1}$. $N_{\mathrm{bare}}$ is the total number of experimental repetitions, and $\sigma_{\mathrm{bare}}$ is the standard deviation of the corresponding $T_1$ data. The evaluated \textit{Gain} is plotted for the three sets in Fig. \ref{fig:gain}.

To model the gain theoretically, we compute the SNR per shot as follows 
\begin{align}\label{eq:th_SNR}
\mathrm{\overline{SNR}}_{\text{bare}} = \frac{F_{\text{bare}}}{\sigma_{\text{bare}}}, ~~ \mathrm{\overline{SNR}}_{\text{QEC}} = \frac{F_{\text{QEC}} \sqrt{p_{\text{success}}}}{\sigma_{\text{QEC}}}
\end{align}
Here, $p_{\mathrm{success}}$ is the probability of successful recovery, and the standard deviations are calculated assuming an underlying binomial distribution realised after qubit measurement. In our calculation, we only account for the measurement error. Thus, we replace the fidelity term $F$ in Eq. \eqref{eq:th_SNR} with $F^* = F(1-E_{\mathrm{meas}}) + (1-F)E_{\mathrm{meas}}$, where $E_{\mathrm{meas}}$ is the average measurement error for the set of investigated qubits, which is of the order of $10^{-2}$, consistent with the data provided by IBM Quantum~\cite{IBM_Torino}.

Hence, the gain of the QEC protocol, obtained theoretically, can be expressed as
\begin{equation}\label{eq:gain}
    \text{Gain}_{th} = \frac{\mathrm{\overline{SNR}}_{\text{QEC}}}{\mathrm{\overline{SNR}}_{\text{bare}}} = \frac{F^*_{\text{QEC}}\sigma_{\text{bare}} }{F^*_{\text{bare}}\sigma_{\text{QEC}}} \sqrt{p_{\text{success}}}
\end{equation}

The theoretically expected gains corresponding to the studied set of qubits are also shown in Fig. \ref{fig:gain}. Owing to a relatively large measurement error, the predicted gain curve goes beyond the break-even mark of 1 only in certain evolution regimes. Furthermore, the model indeed captures the overall trend of the experimental data, though unaccounted gate errors and crosstalk result in excessive deterioration observed in the actual dataset. Nonetheless, being able to observe near-unity gain after implementing a noise-adapted QEC protocol across multiple sets of real-world qubits is a noteworthy achievement.  

\begin{figure}
    \centering
    \includegraphics[width=1\linewidth]{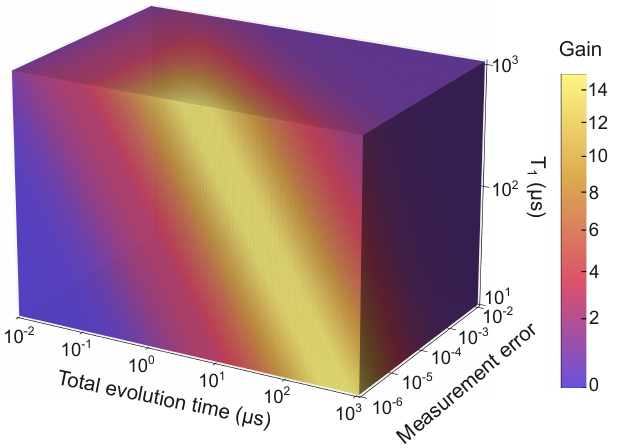}
    \caption{\textbf{Gain estimate on future quantum computers}. The 3D figure shows a simulation of the gain as a function of $T_{1}$, measurement error and delay for $\ket{1_{L}}$. We set $T_{2} = 2T_{1} \mu s$ throughout the process.}
    \label{fig:Gain_3D}
\end{figure}

\section{\label{sec:level4} Summary and Outlook} 

Our work represents an important development in the domain of quantum error correction (QEC), by demonstrating that it is possible to protect against the native, dominant sources of noise on the physical qubits -- amplitude damping and crosstalk -- on a publicly accessible quantum processor and demonstrate multiple logical qubits with break-even fidelities. 

To this end, we employ a noise-adapted, probabilistic QEC scheme that uses only five qubits to protect against single-qubit relaxation -- $3$ physical qubits to encode a single logical qubit, along with $2$ ancillary qubits for syndrome extraction and recovery. In principle, the protocol could be implemented with just one ancilla for both syndrome extraction and recovery, reducing the total qubit count to $4$, but the present choice is motivated by greater conceptual clarity and ease of circuit implementation.
We have incorporated variational quantum circuits to design hardware-efficient encoder and recovery circuits. We have further reduced the complexity of our recovery by using standard decomposition techniques and approximating it without hampering it's ability to correct for first order AD errors.

Due to the low circuit depths of our encoding and recovery circuits, we are able to perform multi-cycle QEC on the hardware and show significant enhancement of the $T_{1}$ lifetimes of multiple logical qubits. To tackle crosstalk, we concatenate our protocol with dynamical decoupling (CHaDD) and validate the methodology on the logical $\ket{+_{L}}$ state. Despite some early evidence of the benefits of combining QEC with dynamical decoupling (DD)~\cite{Boulant_2002_Experimental_QEC_Decoupling_Concatenation}, the question of whether and how DD helps in QEC implementations remains an active area of research~\cite{QEC+DD_Lidar}. Our analysis and experimental results clearly highlight for the first time the pros and cons of combining DD with a nosie-adapted QEC scheme tailored for damping noise. 

Finally, to quantify the advantage of the QEC protocol, we define a measure of gain and show that break-even performance is indeed possible on the currently available NISQ devices, despite the inherent probabilistic nature of the protocol. On the positive side, our protocol is robust to gate error buildup and is limited only by the error in the measurement. Fig. (\ref{fig:Gain_3D}) predicts the gain on future quantum computers. As the measurement errors reduce, we expect to observe significant gain in performance. 

One natural next step is to demonstrate logical gates using the $3$-qubit code. The existence of a transversal non-Clifford logical ($T$) gate is an interesting feature of this code~\cite{Dutta_2024}. An immediate question is whether this can be used to demonstrate non-trivial logical circuits that are resilient to the native noise in the hardware, and thus provide a pathway to demonstrating quantum advantage in the NISQ-era. 

We can also further improve the post-selection procedure by using field programmable gate arrays (FPGAs). This could enable us to terminate the unsuccessful instances during runtime and re-initialize the protocol, thus  reducing the time required for per shot. To maintain uniform experiment times this would increase the total number of shots thereby improving the SNR of the QEC protocol. 

Our work clearly highlights the mature state of quantum hardware today, in that it is possible to run a tailored, resource-efficient QEC protocol via remote access on publicly available quantum machines, in a faithful manner. Having addressed the dominant noise in the hardware, an important future direction is to explore whether such a noise-adapted QEC protocol can be concatenated with general-purpose schemes, relying on code switching strategies, to realize robust, fault-tolerant quantum computers in the future.
\\

\begin{acknowledgments}
This research was supported in part by a grant from the Mphasis F1 Foundation to the Centre for Quantum Information, Communication, and Computing (CQuICC). AR and PM acknowledge funding from the National Quantum Mission, Department of Science and Technology, Gov of India, via grant no. DST/QTC/NQM/QC/2024/1. We acknowledge the use of IBM Quantum in this work. The views expressed were
those of the authors and do not reflect the official policy or position of
IBM or the IBM Quantum team.
\end{acknowledgments}

\section*{Data Availability} \label{data}
The data and the code required to recreate the results in the manuscript are openly available in the link: \href{https://github.com/notSAKRI/Probabilistic-Quantum-Error-Correction/tree/main}{https://github.com/notSAKRI/Probabilistic-Quantum-Error-Correction}.

\bibliography{apssamp}

\newpage

\appendix

\begin{widetext}

\section{State Fidelity and success probability of the $3$-qubit code under damping and dephasing}\label{sec: appA}
\begin{figure}[h!]
    \centering
    \includegraphics[width=0.6\linewidth]{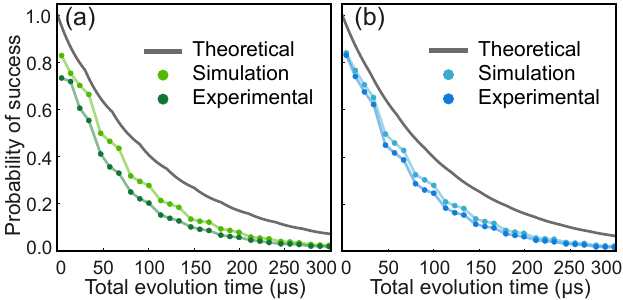}
    \caption{\textbf{The success probability of the multi-QEC protocol.} The green (blue) dots represent the success probability of $\ket{0_{L}}$ ($\ket{1_{L}}$), with the grey solid line being the analytically calculated probability. This comprises of the theoretical, noisy simulation (without crosstalk), and the experimental success probability.}
    \label{fig:app_prob_success}
\end{figure}
Here we write down the analytical expressions for the logical fidelities achieved by the $3$-qubit code, under the combined effect of damping and dephasing. The fidelity expressions for $|0_{L}|\rangle$,  $|1_{L}\rangle$ and $|\pm_{L}\rangle$ have already been given in Sec.~\ref{sec:dephasing}. The fidelity for an arbitrary logical state $|\psi_{L}\rangle$ can be calculated to be, 

 \begin{align}
   F_{\ket{\psi_{L}}}^{ideal} &= 1-\left(\frac{4}{3} \sin ^2\theta  \right) p -\left(\frac{4}{9} \sin ^2\theta (4 \cos \theta +1)\right) p^2 + \left(\frac{1}{3} (2 \cos \theta -1) \sin^2\theta \right) p \gamma 
   + \left(\frac{1}{8} (3 \cos \theta -5) \sin ^2\frac{\theta }{2} \right) \gamma^2 + \text{h.o.}  \\
&= 1 - \left(\frac{(2T_{1} - T_{2})\text{sin}^{2}\theta}{3 T_{1} T_{2}}\right)t - \left(\frac{(45T_{2}^{2} + (32 T_{1}^{2} - 56 T_{1}T_{2} - 7 T_{2}^{2})\text{cos}\theta - 4(2T_{1} - T_{2}) (-4T_{1} + 5T_{2})\text{cos}2\theta) \text{sin}^{2}\frac{\theta}{2}}{72 T_{1}^{2} T_{2}^{2}} \right) t^{2} + \cO(t^{3})
\end{align}

These expressions get modified as follows when we use the approximate recovery instead of the ideal recovery.  
    \begin{align}
        F_{\ket{\psi_{L}}}^{approx} &= 1-\left(\frac{4}{3} \sin ^2\theta\right)p - \left( \frac{4}{9} \sin ^2\theta (4 \cos \theta +1)\right)  p^2 + \left(\frac{4}{3}  \cos \theta  \sin^2 \theta \right) p\gamma + \left(\frac{1}{2}  (\cos \theta - 1)\right) \gamma ^2 + \text{h.o.} \\
        &= 1 - \left(\frac{(2T_{1} - T_{2})\text{sin}^{2}\theta}{3 T_{1} T_{2}} \right)t - \left(\frac{(6T_{2} (-T_{1} + 2T_{2}) + (2T_{1} - T_{2}) ((2T_{1} - 7T_{2})\text{cos}\theta + 2 (T_{1} - 2T_{2}) \text{cos}2\theta)) \text{sin}^{2}\frac{\theta}{2}}{9 T_{1}^{2} T_{2}^{2}} \right) t^{2}+ \cO(t^{3})
    \end{align}

Since $T_2 \leq 2T_1$, the coefficient in-front of time, $t$, is always positive \cite{Krantz_2019_quantum_engineer}. In IBM systems, the qubits' $T_{2}$ lifetimes are $T_{1}$ limited so the coefficient is very small \cite{Gupta_2025}.

As the recovery is subjected to post-selection, there is a finite probability that the protocol will not be successful, and we have to discard the outcome.  When subjected to amplitude-damping noise and the dephasing noise, the probability of successful recovery of the three-qubit code is given by
\begin{align}
p_{success} &=   (1 - \gamma)^{2} \left(1 + \gamma^2 \sin^2 \frac{\theta}{2} + \frac{8}{3}p(p - 1)(1 - \gamma) \cos^2\frac{\theta}{2} \right).
\end{align}

Here $\gamma$ and $p$ are the strength of the AD and dephasing noise. The $\ket{0_{L}}$ state gives the lower bound of the success probability with the expression
\begin{align}
p_{\text{success}} = (1-\gamma)^2 \left(1 + \frac{8p}{3}\left(\gamma + p - 1 - \gamma p \right)\right).
\end{align}
For multiple rounds of quantum error correction, the final state will be error corrected if all the instances of recovery are successful. 

Fig. \ref{fig:app_prob_success} shows the success probability of multi-round QEC with total evolution time, where the recovery is applied after a maximum delay of $30 \mu s$ interval. The theoretical curve accounts for noise affecting only the qubits themselves, while the simulation curve includes all qubit noise, gate noise, and measurement noise. The experimental curve includes crosstalk and other uncertainties, which are difficult to model.

\begin{figure*}[h!]
    \centering
    \includegraphics[width=1.0\linewidth]{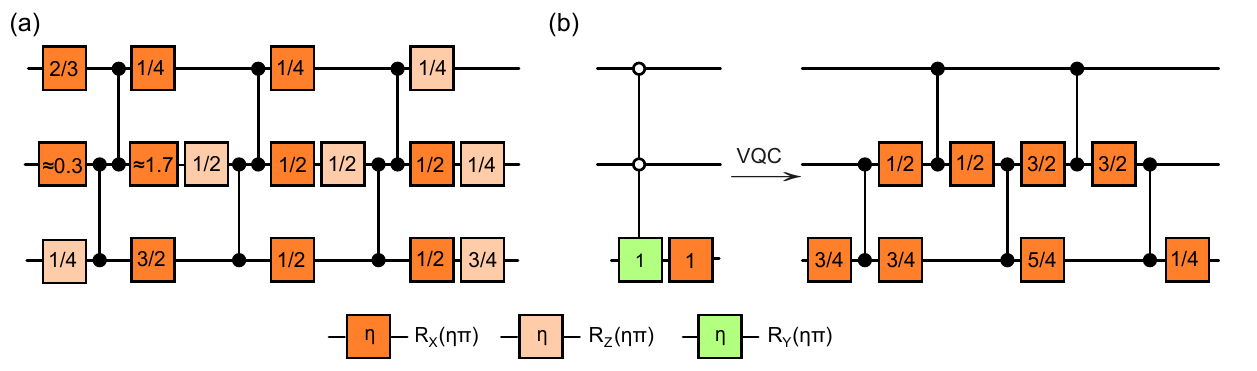}
    \caption{\textbf{Explicit circuits for the components of the recovery}. (a) The transpiled circuit diagram for $U$ as mentioned in Fig. \ref{fig:2}. (b) Transpiled circuit diagram of the last three qubits of the four-qubit unitary $D$, as mentioned in Fig. \ref{fig:2}. As the first qubit does not undergo any gate operations, we have not shown it in the figure. $R_X(\theta)$ is the orange gate, $R_Z(\theta)$ is the peach gate, and $R_Y(\theta)$ is the light green gate. The displayed value in the box, $\eta$, represents the optimized rotational angle as multiples of $\pi$.
    }
    \label{fig:app1}
\end{figure*}

\section{Decomposition of recovery circuit}\label{sec: appB}

In our work, we use the `hardware-efficient' ansatz \cite{Hardware-efficient_ansatz_Kandala2017} to construct the parametrized quantum circuit used in the variational optimization procedure. 
We use IBM Torino whose native gate set is given by $\{CZ, R_X(\theta), R_Z(\theta), R_{ZZ}(\theta), \sqrt{X}, X\}$ and the identity gate; our ansatz is constructed exclusively from these gates.
The ansatz structure consists of alternating layers of parameterized single-qubit rotations and native hardware entangling gates. Additionally, they are arranged while keeping the hardware's connectivity in mind. This structure minimizes circuit depth and gate overhead while preserving sufficient expressibility.
In Fig. \ref{fig:2} (c), we show the quantum circuit for our approximate recovery operation, which includes the gates $V$, $U$, and $D$. We use VQC to find the circuit with the minimum number of two-qubit gates, as shown in Fig. \ref{fig:app1}(a). 

The circuit for the unitary $D$ is shown in Fig. \ref{fig:app1}(b). It appears that the circuit consists of only three qubits, where $D$ is a four-qubit gate. The reason for this is that the first qubit of the four-qubit operation $D$ does not undergo any non-trivial operation under the approximation $\gamma=0$. Hence, $D$ effectively becomes a three-qubit operation which consists of only 5 two-qubit CZ gates.

\section{Optimizing the delay between consecutive QEC rounds}\label{appendix:C}

In the context of multi-QEC protocol, it's important to consider the time interval between successive QEC cycles. Applying QEC too frequently can significantly reduce the overall success probability and accumulate gate errors, while excessively long intervals lead to substantial fidelity loss as it exceeds the code's capability to correct. In addition, to reduce the circuit depth of our recovery operation, we use the approximation that $\gamma = 0$. For small values of $\gamma$, the fidelity of the approximate version is comparable to that of the ideal case. Thus, we need to select an appropriate time interval such that the effective $\gamma$ between two cycles is not too large.

To determine an appropriate interval, we implement the QEC protocol for the $\ket{1_L}$ state at different time intervals, $10 \mu s$, $30 \mu s$, $50 \mu s$, and $70 \mu s$ as shown in Fig. \ref{fig:app2}. As we can see, more frequent QEC cycles yield higher fidelity. However, performing QEC every $10\mu s$ results in a rapid decrease in the success probability. Therefore, we select a $30\mu s$ interval as a balanced choice, whose fidelity is slightly lower than that for $10\mu s$, but noticeably higher than that for $50\mu s$ and $70\mu s$. Longer delays like $50 \mu s$ and $70 \mu s$, result in larger values of $\gamma$ where our approximation then fails.

\begin{figure*}[h!]
    \centering
    \includegraphics[width=1.0\linewidth]{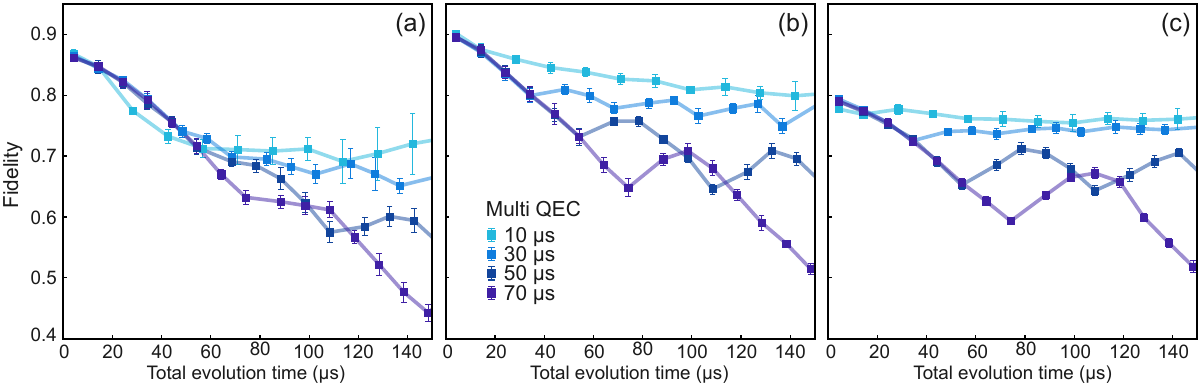}
    \caption{\textbf{Performance of multi-QEC protocol of logical $\ket{1_L}$, for different maximum delay time between QEC cycles, on IBM-Torino}. We can observe that performing QEC more times results in better fidelity, which is somewhat expected; however, it results in a lower probability of success.  After careful consideration, we chose a maximum delay time of $30 \mu s$, as it provides similar performance to $10 \mu s$ compared to the other two cases. Additionally, we can observe a decrease in the standard deviation of the data points, indicating that the probability of success is significantly better than $10 \mu s$. The multi-QEC protocol implemented here does not include the CHaDD scheme.}
    \label{fig:app2}
\end{figure*}

\section{Fidelity expression calculation}\label{Appendix: fidelity}
In this section, we formally go over the steps to show how the $3$-qubit QEC circuit in Fig.~\ref{fig:1} directly leads to an estimate of the fidelity for a given input state, after the final measurement of the logical qubits. We state the result for a general encoding and recovery scheme, since this method of calculating fidelity works for any QEC protocol.
\begin{lemma}[State Fidelity]
For a quantum error correction channel of the form, $\mathcal{N} = \mathcal{U}^{-1} \circ \mathcal{R} \circ \mathcal{E} \circ \mathcal{U}$, where $\mathcal{U}$ is some unitary state preparation and encoding map, $\mathcal{E}$ being the error channel, and $\mathcal{R}$ the corresponding recovery operation, the fidelity of the final state, $\sigma = \mathcal{N}(\rho)$, for some input state $\rho = \ket{0}\bra{0}^{\otimes (n + a)}$, where $n$ is the number of qubits which form the logical qubits and $a$ is the number of ancilla used to implement $\mathcal{R}$, is obtained by measuring the probability of the all zero outcome on the $n$ qubits conditioned on the successful implementation of $\mathcal{R}$.
\end{lemma}

\begin{proof}
We first encode our state by applying a unitary $U_{en}$ giving

\begin{equation*}
    \rho = U_{en}\rho_{Q}U_{en}^{\dagger} \otimes \ket{00}\bra{00}_{a_{1}a_{2}} 
\end{equation*}

The qubits in $\rho_{Q} = \left(\ket{\psi00}\bra{\psi00}\right)_{Q}$ will form our logical qubit protecting $\ket{\psi}$, and $\{a_{1}, a_{2}\}$ are the ancillae used to perform the syndrome and recovery respectively. Now, the AD noise acts on data qubits via a \textit{completely positive trace preserving} (CPTP) map $\mathcal{E}$, defined via a set of Kraus operators $\{E_{k}\}$. We apply recovery operations via another CPTP map with Kraus operators $\{R_{l}\}$. The density matrix after the action of $\mathcal{E}$ is

    \begin{align}
    \rho_{noisy} = \mathcal{E} (\rho) &= \sum_{k} (E_{k} \otimes \idmatrix_{2 \times 2}^{\otimes 2})(\rho)(E_{k}^{\dagger} \otimes \idmatrix_{2 \times 2}^{\otimes 2}) = \sum_{k} E_{k}U_{en}\rho_{Q}U_{en}^{\dagger}E_{k}^{\dagger} \otimes \ket{00}\bra{00}_{a_{1}a_{2}}
\end{align}where, $\idmatrix_{2 \times 2}$ is the identity matrix corresponding to a qubit. Now, we apply the $\mathcal{R}$ map. It consists of the syndrome measurement and the recovery operation controlled on the syndrome outcomes. In the syndrome step, we are performing the $ZZZ$ measurement on our $\rho_{noisy}$. The circuit involves applying \textit{CNOT} gates, where the control is on the data qubits and the target is the syndrome qubit. This gives us the density matrix as

\begin{align}
    \rho_{syndrome} &= \sum_{k_{1}} E_{k_{1}}U_{en}\rho_{Q}U_{en}^{\dagger}E_{k_{1}}^{\dagger} \otimes \ket{1}\bra{1}_{a_{1}} \otimes \ket{0}\bra{0}_{a_{2}} + \sum_{k_{2}} E_{k_{2}}U_{en}\rho_{Q}U_{en}^{\dagger}E_{k_{2}}^{\dagger} \otimes \ket{0}\bra{0}_{a_{1}} \otimes \ket{0}\bra{0}_{a_{2}}
\end{align}

Here, $k_{1}$ captures the no damping and double damping errors, and $k_{2}$ captures the cases of single damping and triple damping on the data qubits respectively. Now, we apply the non-unitary recovery operations, which gives the output as

    \begin{align}
    \rho_{recovered} &= \sum_{k_{1}} \Bigg[R_{0}E_{k_{1}}U_{en}\rho_{Q}U_{en}^{\dagger}E_{k_{1}}^{\dagger}R_{0}^{\dagger} \otimes \ket{1}\bra{1}_{a_{1}} \otimes \ket{0}\bra{0}_{a_{2}} + \bigg(\sqrt{\idmatrix - R_{0}^{\dagger}R_{0}}\bigg)E_{k_{1}}U_{en}\rho_{Q}U_{en}^{\dagger}E_{k_{1}}^{\dagger}\left(\sqrt{\idmatrix - R_{0}^{\dagger}R_{0}}\right)^{\dagger} \otimes \ket{1}\bra{1}_{a_{1}}  \nonumber \\
    & \otimes \ket{1}\bra{1}_{a_{2}}\Bigg] + \sum_{k_{2}} \Bigg[R_{1}E_{k_{2}}U_{en}\rho_{Q}U_{en}^{\dagger}E_{k_{2}}^{\dagger}R_{1}^{\dagger} \otimes \ket{0}\bra{0}_{a_{1}} \otimes \ket{0}\bra{0}_{a_{2}} + \bigg(\sqrt{\idmatrix - R_{1}^{\dagger}R_{1}}\bigg)E_{k_{2}}U_{en}\rho_{Q}U_{en}^{\dagger}E_{k_{2}}^{\dagger}\left(\sqrt{\idmatrix - R_{1}^{\dagger}R_{1}}\right)^{\dagger} \nonumber\\
    &\otimes \ket{0}\bra{0}_{a_{1}} \otimes \ket{1}\bra{1}_{a_{2}}\Bigg]
    \end{align}

After tracing out the syndrome qubits, we get,
 
    \begin{align}
    \tilde{\rho}_{recovered} &= \text{Tr}_{a_{1}}\left[\rho_{recovered}\right] \nonumber \\ 
    &= \left[\sum_{k_{1}} R_{0}E_{k_{1}}U_{en}\rho_{Q}U_{en}^{\dagger}E_{k_{1}}^{\dagger}R_{0}^{\dagger} + \sum_{k_{2}} R_{1}E_{k_{2}}U_{en}\rho_{Q}U_{en}^{\dagger}E_{k_{2}}^{\dagger}R_{1}^{\dagger} \right] \otimes \ket{0}\bra{0}_{a_{2}}  + \Bigg[\sum_{k_{1}} \bigg(\sqrt{\idmatrix - R_{0}^{\dagger}R_{0}}\bigg)E_{k_{1}}U_{en}\rho_{Q}U_{en}^{\dagger}E_{k_{1}}^{\dagger}\nonumber \\
    &\left(\sqrt{\idmatrix - R_{0}^{\dagger}R_{0}}\right)^{\dagger} + \sum_{k_{2}} \left(\sqrt{\idmatrix - R_{1}^{\dagger}R_{1}}\right)E_{k_{2}}U_{en}\rho_{Q}U_{en}^{\dagger}E_{k_{2}}^{\dagger}\left(\sqrt{\idmatrix - R_{1}^{\dagger}R_{1}}\right)^{\dagger} \Bigg] \otimes \ket{1}\bra{1}_{a_{2}}
\end{align} 
The next step is the post-selection step. Here, our protocol is successful for the $\ket{0}$ output of ancilla $a_{2}$. So, our final density matrix, after normalization, is,
 
    \begin{align}
    \tilde{\sigma} = \mathcal{R} \circ \mathcal{E}(\rho) &= \frac{\sum_{k_{1}} R_{0}E_{k_{1}}U_{en}\rho_{Q}U_{en}^{\dagger}E_{k_{1}}^{\dagger}R_{0}^{\dagger} + \sum_{k_{2}} R_{1}E_{k_{2}}U_{en}\rho_{Q}U_{en}^{\dagger}E_{k_{2}}^{\dagger}R_{1}^{\dagger}}{\text{Tr}\left[\sum_{k_{1}} R_{0}E_{k_{1}}U_{en}\rho_{Q}U_{en}^{\dagger}E_{k_{1}}^{\dagger}R_{0}^{\dagger} + \sum_{k_{2}} R_{1}E_{k_{2}}U_{en}\rho_{Q}U_{en}^{\dagger}E_{k_{2}}^{\dagger}R_{1}^{\dagger}\right]} \\ \nonumber
    &= \frac{\sum_{(l, k) \in W}R_{l}E_{k}U_{en}\rho_{Q}U_{en}^{\dagger}E_{k}^{\dagger}R_{l}^{\dagger}}{\text{Tr}\left[\sum_{(l, k) \in W}R_{l}E_{k}U_{en}\rho_{Q}U_{en}^{\dagger}E_{k}^{\dagger}R_{l}^{\dagger}\right]}
    \end{align} 

The set $W = \{(l=0, k=k_{1})\} \bigcup \{(l=1, k=k_{2})\}$ represents pairs of labels of the Kraus operators of the recovery and error channels, respectively. The state fidelity becomes,

    \begin{align}\label{fidelity_expression}
    F^{2} &= \text{Tr}\left[\left(U_{en}\rho_{Q}U_{en}^{\dagger}\right)\tilde{\sigma} \right] \\\nonumber
    &= \frac{\text{Tr}\left[\left(U_{en}\rho_{Q}U_{en}^{\dagger}\right) \left( \sum_{(l, k) \in W}R_{l}E_{k}U_{en}\rho_{Q}U_{en}^{\dagger}E_{k}^{\dagger}R_{l}^{\dagger} \right)\right]}{\text{Tr}\left[\sum_{(l, k) \in W}R_{l}E_{k}U_{en}\rho_{Q}U_{en}^{\dagger}E_{k}^{\dagger}R_{l}^{\dagger}\right]} \\ \nonumber
    &= \frac{\sum_{(l, k) \in W}|\bra{0}^{\otimes 3}G^{\dagger}U_{en}^{\dagger}R_{l}E_{k}U_{en}G\ket{0}^{\otimes 3}|^{2} }{\text{Tr}\left[\sum_{(l, k) \in W}R_{l}E_{k}U_{en}\rho_{Q}U_{en}^{\dagger}E_{k}^{\dagger}R_{l}^{\dagger}\right]} \; = \;  \frac{\sum_{(l, k) \in W} |\bra{0}^{\otimes 3}\tilde{U}_{en}^{\dagger}R_{l}E_{k}\tilde{U}_{en}\ket{0}^{\otimes 3}|^{2}}{\text{Tr}\left[\sum_{(l, k) \in W}R_{l}E_{k}U_{en}\rho_{Q}U_{en}^{\dagger}E_{k}^{\dagger}R_{l}^{\dagger}\right]}, \\ \nonumber
    \end{align}
    where, $\ket{\psi} = G\ket{0}$ is prepared using unitary $G$, and $\tilde{U}_{en} = U_{en}G$. The final structure of the density matrix after $U_{en}^{\dagger}$ acts on $\tilde{\sigma}$, in block form, is,
\begin{align}
    \tilde{U}_{en}^{\dagger}\tilde{\sigma}\tilde{U}_{en} = \begin{pmatrix}
        F^{2} & * \\
        * & *
    \end{pmatrix}
\end{align}
In the top-left, $F^{2}$, represents the fidelity expression, Eq. (\ref{fidelity_expression}), corresponding to the all zero outcome of the circuit. It is clear that the above fidelity calculation protocol can be generalized to any quantum error correction scheme.
\end{proof}
\end{widetext}
\end{document}